\documentclass[12pt]{article}

\usepackage[english]{babel}
\usepackage[T1]{fontenc}
\usepackage[utf8]{inputenc} 
\usepackage[final]{microtype}

\usepackage{amsmath,amsfonts,amssymb,latexsym}

\usepackage{graphicx}
\usepackage{caption}
\usepackage{subcaption}
\usepackage{float} 

\usepackage{hyphenat}

\allowdisplaybreaks[2]

\newtheorem{theorem}{Theorem}[section]
\newtheorem{corollary}{Corollary}[section]

\newtheorem{proposition}{Proposition}[section]
\newtheorem{remark}{Remark}[section]
\newtheorem{example}{Example}[section]
\newenvironment{proof}[1][Proof]{\textsc{#1.} }{\ \rule{0.5em}{0.5em}}
\numberwithin{equation}{section}

\def\be{\begin{equation}}
\def\ee{\end{equation}}
\def\bq{\begin{eqnarray}}
\def\eq{\end{eqnarray}}
\def\beq{\begin{eqnarray}}
\def\eeq{\end{eqnarray}}

\usepackage{newunicodechar}
\newunicodechar{Ω}{\ensuremath{\Omega}}
\newunicodechar{μ}{\ensuremath{\mu}}
\newunicodechar{σ}{\ensuremath{\sigma}}
\newunicodechar{ν}{\ensuremath{\nu}}
\newunicodechar{ρ}{\ensuremath{\rho}}
\newunicodechar{γ}{\ensuremath{\gamma}}
\newunicodechar{≤}{\ensuremath{\le}}
\newunicodechar{≥}{\ensuremath{\ge}}
\newunicodechar{≈}{\ensuremath{\approx}}
\newunicodechar{≠}{\ensuremath{\ne}}
\newunicodechar{→}{\ensuremath{\to}}
\newunicodechar{×}{\ensuremath{\times}}
\newunicodechar{−}{\ensuremath{-}} 

\newunicodechar{–}{\textendash}
\newunicodechar{—}{\textemdash}

\usepackage[hidelinks]{hyperref}
\providecommand{\texorpdfstring}[2]{#1}
\begin{document}

\title{\textsc{Cosmic acceleration as a saddle-node bifurcation: background identities and structure}}

\author{\Large{\textsc{Spiros Cotsakis}}$^{1,2}$\thanks{\texttt{skot@aegean.gr}}\\
$^{1}$Clare Hall, University of Cambridge, \\
Herschel Road, Cambridge CB3 9AL, United Kingdom\\  \\
$^{2}$Institute of Gravitation and Cosmology,  RUDN University\\
ul. Miklukho-Maklaya 6, Moscow 117198, Russia}
\date{January 2026}
\maketitle
\newpage

\begin{abstract}
\noindent We show that the late-time acceleration of the universe can be understood as a codimension-one bifurcation of the Friedmann dynamical system in the variables $(H,\Omega)$. At a critical value of the density-parameter combination, a saddle-node bifurcation occurs; beyond the saddle-node, trajectories are globally attracted to a new accelerating fixed point. We obtain a normal form and a versal unfolding for the reduced dynamics, proving robustness (structural stability) of the phenomenon and deriving the characteristic square-root splitting of the emerging equilibria. We interpret the unfolding parameter as a measure of departure from adiabaticity via a modified continuity/entropy balance, thus linking acceleration to controlled non-equilibrium evolution rather than to a cosmological constant. In particular, late-time acceleration arises without invoking a separate dark-energy fluid; it emerges from a bounded unfolding of the background flow around a saddle–node organizing center. We situate this within a broader “general-relativity landscape,” where control parameters act as moduli and branches of exact solutions appear as equilibrium loci, allowing bifurcation-theoretic tools to organize cosmological dynamics without introducing extra fields, and suggesting a coherent, bifurcation-guided cosmic history.
\end{abstract}
\newpage
\tableofcontents
\newpage

\section{Introduction}\label{sec:intro}
Modern developments in gravitational physics have revealed deep parallels between the solution spaces of general relativity and the string–theory landscape \cite{cot24c}. Recasting Einstein’s equations in the language of versal unfoldings and bifurcation theory exposes a structured moduli space of gravitational equilibria: control parameters (equation-of-state data, matter couplings, symmetry assumptions) act as moduli, while branches of exact solutions (Schwarzschild, Friedmann, Kerr, etc.) appear as equilibrium loci. Bifurcation surfaces mark the creation, merger, or annihilation of branches under parameter variations.

The resulting \emph{general-relativity landscape} admits tools from singularity and catastrophe theory (normal forms, versal unfoldings, bifurcation diagrams) that catalogue gravitational instabilities and mode interactions (saddle–node, pitchfork, Hopf, homoclinic, higher-order interactions) \cite{cot23a,cot23b,cot24a}. This suggests a moduli-space organising principle for nonlinear field theories and provides a roadmap for numerical exploration and for observational tests of novel black-hole phases or cosmological attractors.

Cosmic acceleration is a central, well-established phenomenon \cite{a1,a2}. Despite many proposals—dark energy \cite{d1,d2}, inhomogeneities \cite{i1,i2}, quintessence \cite{q1,q2}, modified gravity \cite{m1,m2}—no single explanation is definitive (see \cite{a3}). In \cite{cot23} it was shown that, for the simplest Friedmann cosmologies (without \(\Lambda\)), there exists a uniquely defined perturbation of the Friedmann equations that meets the conditions for a saddle–node bifurcation: near the Milne and flat states, the unfolded dynamics develops a pair of equilibria that organise nearby geometries.

In this paper we prove that the reduced Friedmann flow admits a versal extension with normal form \(Z'=\bar{\mu}-Z^{2}\), identify the accelerating future attractor beyond the saddle–node, and give a controlled non-equilibrium interpretation of the unfolding parameter via a modified continuity/entropy balance. The mechanism links acceleration to self-organization out of equilibrium and is robust (structurally stable) in the sense of versal unfoldings. We obtain late-time acceleration without introducing a separate dark-energy fluid; instead it emerges from a bounded unfolding of the background flow around a saddle–node organizing center.

The plan is as follows. In Section~\ref{sec:unfolding} we motivate the need for a versal deformation of the Friedmann system, give parallel dynamical and algebraic formulations, and isolate the non-hyperbolic regimes. In Section~\ref{sec:versal} we carry out the centre-manifold reduction and derive the universal unfolding with its fixed branches. In Section~\ref{sec:accel} we interpret the unfolding parameters physically via a modified continuity/entropy balance, and identify the accelerating future attractor. In Section~\ref{sec:discussion} we summarise results, note limitations, and sketch future directions, in particular, we outline observational handles to be developed in a companion paper \cite{cot26a}.

\textbf{Terminology and main results.}
We use “hyperbolic” to mean that all equilibria are hyperbolic, i.e. every eigenvalue of the linearized Jacobian at each fixed point has nonzero real part. Such systems are structurally stable in the sense that their phase portraits are preserved under small perturbations (cf.\ the Hartman–Grobman theorem).

By a \emph{versal unfolding} we mean a codimension-one deformation that captures all small perturbations of a degenerate equilibrium up to smooth changes of variables, parameters, and time. For the reduced Friedmann density equation we work with
\[
\Omega'=\Omega(\Omega-1)+\nu,\qquad \nu=\sigma/\mu,
\]
which, under the shift \(Z=\tfrac12-\Omega\) and \(\bar\mu=\tfrac14-\nu\), takes the saddle–node normal form
\[
Z'=\bar\mu - Z^2.
\]
Here \(\mu=3\gamma-2\), $\gamma$ being the fluid parameter,  and \(\sigma\) is the unfolding parameter (related to entropy production). The saddle–node point is \((Z,\bar\mu)=(0,0)\), i.e.\ \((\Omega,\nu)=(\tfrac12,\tfrac14)\).

More explicitly, our contributions are:
\begin{enumerate}
  \item Derivation of the saddle–node normal form for the reduced Friedmann flow,
  \(Z'=\bar{\mu}-Z^{2}\), together with the map \((\mu,\sigma)\mapsto(\nu,\bar{\mu})\) and the stability of the fixed branches \(\Omega_\pm\).
  \item A non-equilibrium interpretation of the unfolding via
  \(\dot{\rho}+3H(\rho+p)=3\sigma H^{3}\) and \(\dot{S}/V=(3\sigma/T)H^{3}\), including the conditions for late-time acceleration at the future attractor.
  \item A macroscopic interpretation of the non-adiabatic source term, and discuss how departures from the linearized  scaling $\rho\propto a^{-3(1+w)}$ are governed by the unfolding parameter $\nu$ leads to novel probes for non-adiabaticity.

\end{enumerate}

\section{The need for a versal extension of Friedmann cosmology}\label{sec:unfolding}
In this section we set up the Friedmann system in two parallel ways—an autonomous dynamical form in \((H,\Omega)\) and the familiar algebraic/continuity formulation—and show their equivalence in the hyperbolic regime (for background on hyperbolic equilibria and structural stability in cosmological dynamical systems, see \cite{we97}, chs.~4–5). We fix notation, derive the reduced \((H,\Omega)\) equations, and identify the flat/Milne equilibria together with the closed-model incompleteness. This makes transparent where non-hyperbolic behaviour occurs and motivates, in the next section, a versal (codimension-one) extension capable of capturing the bifurcation that underlies acceleration.

We start from the symmetry-reduced Einstein equations for a Friedmann--Robertson--Walker (FRW) metric with scale factor $a(t)$ (dot $=\mathrm{d}/\mathrm{d}t$) and a perfect fluid with equation of state $p=(\gamma-1)\rho$ (equivalently $p=w\rho$ with $w=\gamma-1$). The time--time component (Raychaudhuri),
\begin{equation}\label{ray}
3\,\ddot{a}=-4\pi G(\rho+3p)\,a,
\end{equation}
the space--space component,
\begin{equation}\label{ss}
a\ddot{a}+2\dot{a}^2+2k=4\pi G(\rho-p)a^2,
\end{equation}
and the resulting Friedmann constraint,
\begin{equation}\label{fr}
\frac{\dot{a}^2}{a^2}+\frac{k}{a^2}=\frac{8\pi G}{3}\rho,
\end{equation}
are supplemented by the continuity equation
\begin{equation}\label{cont}
\dot{\rho}+3\frac{\dot{a}}{a}(\rho+p)=0,
\end{equation}
which follows from $\nabla_\nu T^{\mu\nu}=0$.

Two standard routes exist: a \emph{dynamical-systems} formulation and an \emph{algebraic} formulation. We begin with the former.
\subsection{Dynamical systems approach}
One convenient autonomous form is
\begin{equation}\label{bas-1}
\begin{split}
\frac{\mathrm{d}\Omega}{\mathrm{d}\tau}&=-\mu\,\Omega+\mu\,\Omega^2,\\
\frac{\mathrm{d}H}{\mathrm{d}\tau}&=-H-\frac{1}{2}\mu\,\Omega\,H,
\end{split}
\end{equation}
where $H=\dot{a}/a$ is the Hubble parameter, $\Omega=\rho/(3H^2)$ is the density parameter, and $\tau$ is the dimensionless time defined by $\mathrm{d}t/\mathrm{d}\tau=1/H$ (primes below denote $\mathrm{d}/\mathrm{d}\tau$). The combination
\be\label{eq:q-eq}
q=\frac{1}{2}\,\mu\,\Omega,
\ee
is the deceleration parameter.

Equation \eqref{bas-1} is directly equivalent to \eqref{ray}, \eqref{fr}, \eqref{cont}: \eqref{bas-1}a is the continuity equation in these variables, \eqref{bas-1}b is Raychaudhuri; the Friedmann constraint is preserved along the flow. Since its precise value is not known, we treat $\mu$ as a (frozen) parameter: $\mu'=0$.

An advantage of \eqref{bas-1} is that key FRW solutions become \emph{equilibria} that organise the phase portrait. For $\mu>0$ (i.e.\ $\gamma>2/3$) the $\Omega$–equation reads $\Omega'=\mu\,\Omega(\Omega-1)$, so
\be
\Omega'>0\quad(\Omega>1),\qquad
\Omega'<0\quad(\Omega\in(0,1)).
\ee
Thus the flat equilibrium $\Omega=1$ is a repeller (past attractor), and the Milne equilibrium $\Omega=0$ is an attractor (future attractor):
\be
\lim_{\tau\to-\infty}\Omega=1,\qquad
\lim_{\tau\to+\infty}\Omega=0.
\ee
For closed models ($\Omega>1$, $k=+1$) one has
\begin{equation}\label{incom}
\lim_{\tau\to\tau_{\max}}\Omega=+\infty,
\end{equation}
i.e.\ the standard future incompleteness (maximum expansion before recollapse). Observational constraints on $(t_0,H_0,q_0,\Omega_0)$ fit naturally into this phase-space picture (see e.g.\ \cite{we97}, ch.~3; \cite{hob}, ch.~15).

Finally, from Raychaudhuri \eqref{ray} one sees immediately that for $\mu>0$ (i.e.\ $w>-1/3$) the pure FRW models without $\Lambda$ \emph{cannot accelerate}:
\be
\ddot{a}<0 \quad\Longleftrightarrow\quad \mu>0.
\ee

\subsection{Algebraic approach}

A widely used complementary route starts from the differential form of the continuity law on a comoving volume $V=2\pi^2 a^3$,
\begin{equation}\label{cont-diff}
\mathrm{d}(\rho V)=-p\,\mathrm{d}V,
\end{equation}
and, with $p=w\rho$, yields after integrating the \emph{linearised} relation
\be
(1+w)\rho\,\mathrm{d}V+V\,\mathrm{d}\rho=0,
\ee
the familiar scaling,
\begin{equation}\label{ra}
\rho=C\,a^{-3(1+w)}.
\end{equation}
Substituting \eqref{ra} into the (algebraic) Friedmann constraint and treating the Raychaudhuri equation via the deceleration $q$ leads to the ``algebraized'' FRW system used for distance–redshift and related observables (see \cite{weinberg1,hob}).

To see the parallel with the dynamical-systems result \eqref{incom}, write the Friedmann equation as (we shall be using units where $8\pi G=1$, with $k\in\{-1,0,+1\}$ for open/flat/closed),
\begin{equation}\label{a-dot}
\dot{a}=\pm\left(\frac{1}{3}\rho\,a^{2}-k\right)^{1/2},
\end{equation}
so $\dot{a}$ is real $\iff\ \rho\,a^{2}\ge 3k$ (with equality at the turning point $\dot a=0$). With \eqref{ra} (and $C>0$) this gives, for $k=+1$ and $3w+1>0$, the upper bound
\begin{equation}
a<\left(\frac{C}{3k}\right)^{1/(3w+1)},
\end{equation}
again showing future incompleteness for closed models (no such restriction for $k=-1$).\footnote{In the $(a,\rho)$ variables the vector field defined by \eqref{a-dot} is not smooth at the turning point, which is a mathematical reason to prefer formulations like \eqref{bas-1} for linearisation and stability.}

\subsection{Equivalence in the hyperbolic regime}

Both approaches above are based on linearised versions of \eqref{ray}, \eqref{fr}, \eqref{cont} and therefore agree in the \emph{hyperbolic} case, where all relevant equilibria have no centre directions. In that setting, structural stability justifies treating $\gamma$ (and hence $\mu$) as fixed.

\section{The versal deformation of the Friedmann equations}\label{sec:versal}
In this section we give a concise, self-contained derivation of a versal extension of the Friedmann equations. Our aim is twofold: (i) to place parts of the broader literature within a common dynamical-systems framework, and (ii) to prepare the ground for a bifurcation-theoretic interpretation of cosmic acceleration that naturally unifies several phenomenological models. We carry out the centre-manifold reduction near the dispersive equilibria, derive the saddle–node  normal form, and obtain the fixed branches and their stability.

\subsection{Beyond hyperbolicity: non-hyperbolic equilibria}

The linearised (algebraic) relation \eqref{ra} implicitly discards second-order terms in the product $\rho V$. While any smooth function may be truncated locally, the same is not automatically true for solutions of the coupled nonlinear FRW system. Whether a linear truncation suffices depends on intrinsic dynamical features—chiefly the nature of equilibria of \eqref{bas-1}.

If all equilibria are hyperbolic, linearisation captures the qualitative dynamics and small changes of $\gamma$ do not alter stability. The interesting case is when \eqref{bas-1} has a non-hyperbolic equilibrium $(\Omega_*,H_*,\gamma_*)$, i.e.\ the Jacobian has a zero eigenvalue. Then, for $\gamma$ near $\gamma_*$ and $(\Omega,H)$ near $(\Omega_*,H_*)$, the number and type of equilibria may change—new branches can be created or annihilated.

As shown in \cite{cot23}, in addition to the classical hyperbolic Milne $(0,0,\mu_0)$ and flat $(1,0,\mu_0)$ states (with $\mu_0$ fixed), the system \eqref{bas-1} exhibits the following non-hyperbolic equilibria:
\begin{enumerate}
\item \textbf{EQ-I (dispersive Milne):} the $\Omega$–axis $(\Omega,0,0)$ (including the origin),
\item \textbf{EQ-IIa (dispersive flat):} $(1,0,0)$,
\item \textbf{EQ-IIb:} the $H$–axis $(1,H,-2)$, with $H=0,1$ corresponding to Einstein static and de Sitter.
\end{enumerate}
At these points the Jacobian is
\[
J_{\mathrm{EQ\!-\!I}}=J_{\mathrm{EQ\!-\!IIa}}=\mathrm{diag}(0,-1),\qquad
J_{\mathrm{EQ\!-\!IIb}}=\mathrm{diag}(-2,0),
\]
i.e.\ exactly one zero and one nonzero eigenvalue. Linear stability is inconclusive; one must use centre-manifold and unfolding theory. In particular, $\gamma$ can no longer be held fixed but should be treated as a parameter, leading to a one-parameter family of vector fields.

For a self-contained treatment we therefore work within bifurcation and singularity theory (see also the general-relativity landscape programme \cite{cot24c}). Following \cite{cot23}, the versal unfolding proceeds in three steps:
\begin{enumerate}
\item centre-manifold reduction of \eqref{bas-1} near the dispersive Milne state,
\item universal (versal) deformation of the reduced system,
\item analysis of the resulting parametrised equilibria (fixed branches).
\end{enumerate}
Here ``bifurcation'' refers specifically to dynamical bifurcations of equilibria with a one-dimensional centre manifold (one zero eigenvalue).

\subsection{Dynamics on the centre manifold}

Near \textbf{EQ-I} the system can be written as
\begin{equation}\label{sys3a}
\begin{pmatrix}\Omega'\\ H'\end{pmatrix}
=
\begin{pmatrix}0&0\\ 0&-1\end{pmatrix}
\begin{pmatrix}\Omega\\ H\end{pmatrix}
+
\begin{pmatrix}-\mu\,\Omega+\mu\,\Omega^2\\ -\tfrac12\mu\,\Omega\,H\end{pmatrix},
\end{equation}
with
\begin{equation}\label{sys3b}
\mu'=0.
\end{equation}
The (parametric) centre-manifold theorem yields a local centre manifold
\begin{equation}\label{cenman}
W^c_{\mathrm{loc}}=\{(\Omega,\mu,H):\, H=h(\Omega,\mu),\ |\Omega|<\delta_1,\ |\mu|<\delta_2,\ h(0,0)=0,\ Dh(0,0)=0\},
\end{equation}
for $\delta_{1,2}$ small. Using the tangency condition $\dot H-D_\Omega h\,\dot\Omega=0$ and an ansatz
\be
h(\Omega,\mu)=a\Omega^2+b\mu\Omega+c\mu^2+O(3),
\ee
one finds $a=b=c=0$, so $W^c_{\mathrm{loc}}$ is the $\Omega$–axis and the reduced dynamics is
\begin{align}
\Omega'&=-\mu\,\Omega+\mu\,\Omega^2, \label{ooc}\\
\mu'&=0. \label{m0}
\end{align}
Here $-\mu\Omega$ is quadratic and $\mu\Omega^2$ cubic in $(\Omega,\mu)$. The origin is a degenerate equilibrium:
\begin{equation}\label{sn1}
f(0,0)=0,\qquad \partial_\Omega f(0,0)=0,\qquad
\partial_\mu f(0,0)=0,\qquad \partial_{\Omega\Omega} f(0,0)=0,
\end{equation}
so both transversality and nondegeneracy fail.

\subsection{Friedmann cosmology as organising centre}

Degenerate germs like \eqref{ooc} are treated as organising centres. The key questions are:
(i) finite determinacy—do low-order terms determine the qualitative dynamics irrespective of higher-order corrections? and
(ii) unfolding—how many auxiliary parameters suffice to capture all small perturbations?

A natural low-order perturbation is
\begin{equation}\label{un1a}
\Omega'=\mu\,\Omega(\Omega-1)+\sigma,
\end{equation}
i.e.
\be
G(\Omega,\mu,\sigma)=-\mu\Omega+\mu\Omega^2+\sigma.
\ee
For $\sigma=0$ the solution branches $\{\Omega=0\}$ and $\{\Omega=1\}$ are regular away from the origin, since
\be
\frac{\partial G}{\partial\Omega}(\Omega,\mu,0)=-\mu+\mu\Omega\neq 0
\ee
there. By the implicit-function theorem, near each regular point one can solve the equation $G(\Omega,\mu,\sigma)=0$ for $\Omega=\Omega(\sigma)$ smoothly in $\sigma$, producing the familiar parabola of a saddle–node bifurcation. Near the origin, the quadratic term $-\mu\Omega$ dominates the cubic $\mu\Omega^2$, so \eqref{un1a} is well-approximated by the hyperbola $\mu\Omega=\sigma$; at $\sigma=0$ the system is most singular. This viewpoint motivates the versal normal form developed in the next subsection.

\subsection{Normal form and versal unfolding}

We eliminate $\mu$ by rescaling time,
\begin{equation}\label{T}
T=\mu\,\tau,
\end{equation}
and combine the two parameters into the single \emph{unfolding-fluid parameter}
\begin{equation}\label{nu}
\nu=\frac{\sigma}{\mu}.
\end{equation}
The continuity equation becomes,
\begin{equation}\label{un1b}
\frac{\mathrm{d}\Omega}{\mathrm{d}T}=\Omega(\Omega-1)+\nu.
\end{equation}
Here $\nu$ (equivalently $\sigma$) plays the role of the auxiliary or \emph{unfolding} parameter, while $\mu$ (or $\gamma$) is the original distinguished parameter. The limit $\sigma=0$ (i.e.\ $\nu=0$) recovers the degenerate Friedmann system.
\begin{figure}[t]
  \centering
  \IfFileExists{muBarZ-1.pdf}{%
    \includegraphics[width=.47\linewidth]{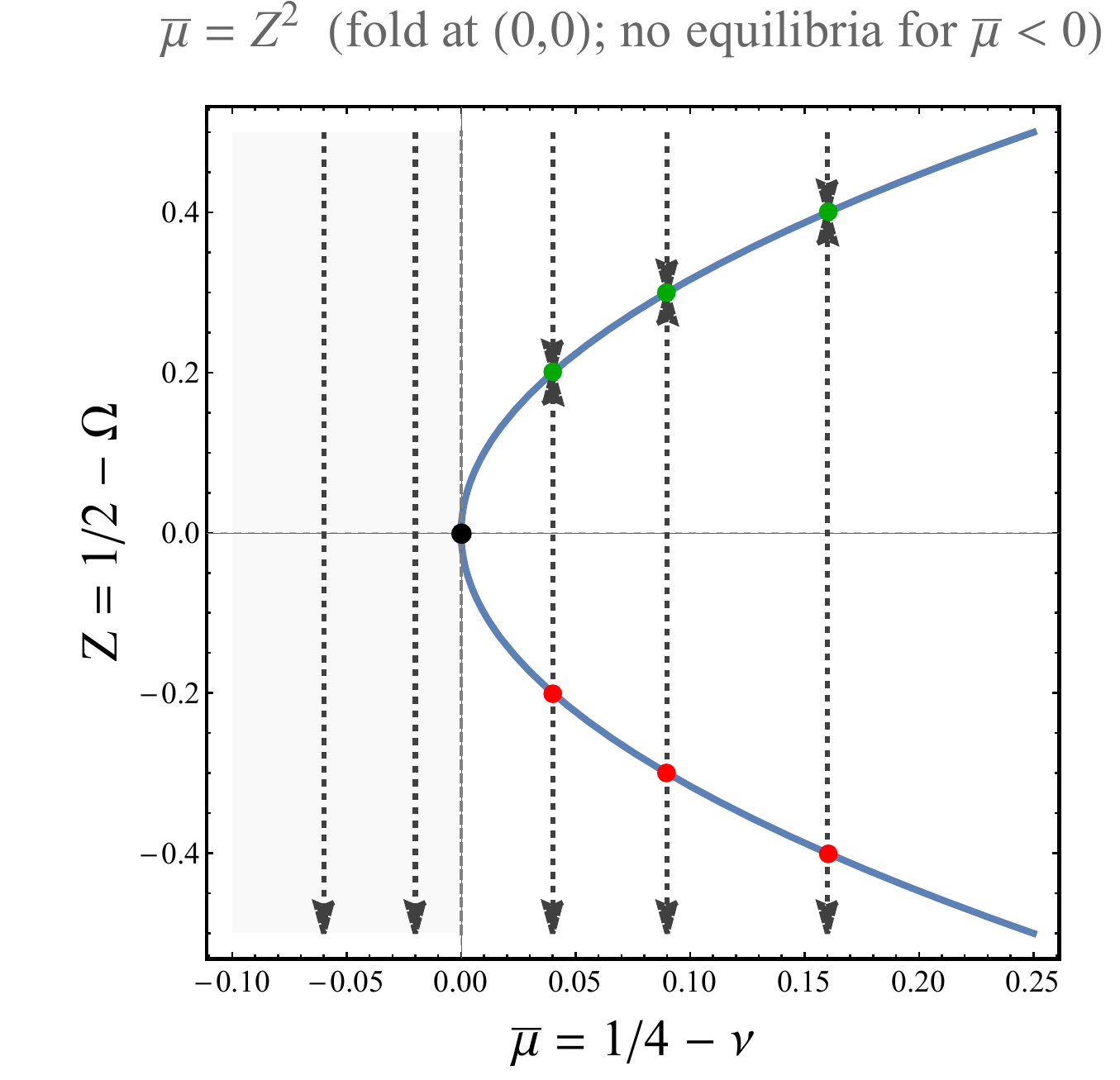}%
  }{%
    \fbox{\parbox[c][.32\textheight][c]{.47\linewidth}{\centering Placeholder: muBarZ-1.pdf}}%
  }\hfill
  \IfFileExists{nuOmega.pdf}{%
    \includegraphics[width=.47\linewidth]{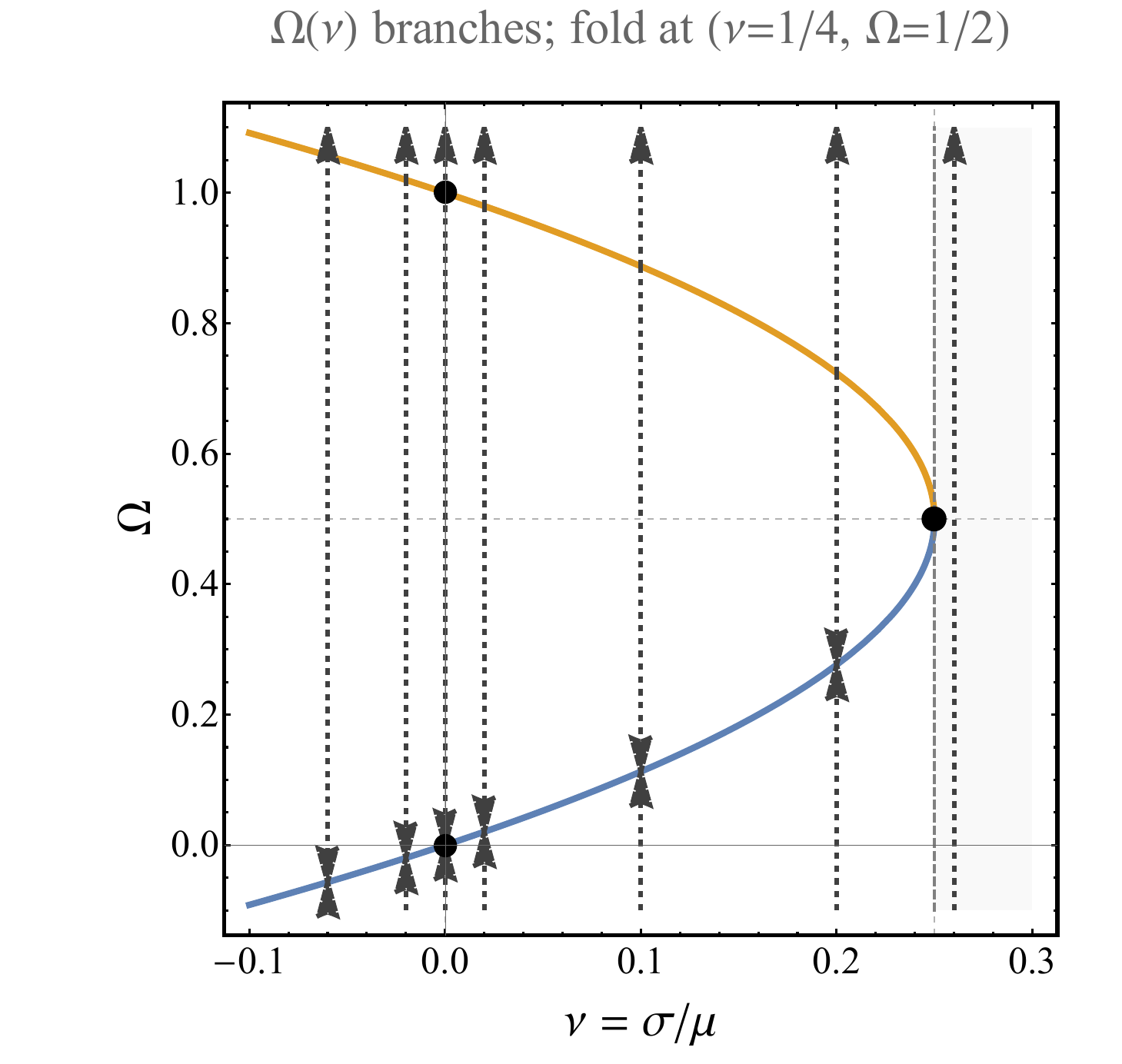}%
  }{%
    \fbox{\parbox[c][.32\textheight][c]{.47\linewidth}{\centering Placeholder: nuOmega.pdf}}%
  }
  \caption{\textbf{Saddle–node structure in shifted and normalized variables.} \emph{Left:} $(\bar{\mu},Z)$ plane with $Z$ vertical. The parabola $\bar{\mu}=Z^{2}$ opens to the right and the saddle-node is at $(0,0)$. On vertical phase lines $Z'=\bar{\mu}-Z^{2}$ the flow is down–up–down for $\bar{\mu}>0$ (stable $Z=+\sqrt{\bar{\mu}}$, unstable $Z=-\sqrt{\bar{\mu}}$); for $\bar{\mu}<0$ there are no real equilibria (shaded). \emph{Right:} $(\nu,\Omega)$ plane with $\Omega$ vertical. The branches $\Omega_{+}(\nu)=\tfrac12-\tfrac12\sqrt{1-4\nu}$ and $\Omega_{-}(\nu)=\tfrac12+\tfrac12\sqrt{1-4\nu}$ exist for $\nu<\tfrac14$ and meet at $(\nu,\Omega)=(\tfrac14,\tfrac12)$. On vertical phase lines $\Omega'=\Omega(\Omega-1)+\nu$ the flow is up–down–up for all $\nu<\tfrac14$, so $\Omega_+$ is stable and $\Omega_-$ unstable; at $\nu=1/4$ the branches coalesce; for $\nu>\tfrac14$ the flow is up everywhere. The case $\nu=0$ (i.e.\ $\bar{\mu}=\tfrac14$) recovers $\Omega=0$ and $\Omega=1$; for $0<\nu<\tfrac14$ these are replaced by $(\Omega_+,\Omega_-)$.} 
  \label{fig:SN-Z-muBar-nu-Omega}
\end{figure}
To reveal the canonical structure, introduce the shifted density
\begin{equation}\label{Z}
Z=\frac{1}{2}-\Omega,
\end{equation}
and the distance to the saddle-node,
\be
\bar{\mu}=\frac{1}{4}-\nu.
\ee
Then \eqref{un1b} takes the \emph{saddle–node normal form} \cite{cot23}
\begin{equation}\label{un3}
\frac{\mathrm{d}Z}{\mathrm{d}T}=\bar{\mu}-Z^2,
\end{equation}
with bifurcation point at $(Z,\bar{\mu})=(0,0)$; see Fig.~\ref{fig:SN-Z-muBar-nu-Omega}. In particular,
\be\label{cond}
g(Z,\bar{\mu})=\bar{\mu}-Z^2,\quad
g=0,\ \partial_Z g=-2Z=0,\ \partial_{ZZ}g=-2\neq0,\ \partial_{\bar{\mu}}g=1\neq0,
\ee
so both nondegeneracy and transversality hold.

\textbf{Codimension and the structural--instability locus.}
In the versal--unfolding framework, ``codimension'' is not meant in the purely geometric sense of an embedded submanifold in physical space.
Rather, one considers the ambient space $\mathcal{X}$ of smooth vector--field germs (modulo smooth changes of state variables), and the subset
\[
S^{c}=\{\, f\in\mathcal{X}:\ f \text{ has a nonhyperbolic equilibrium}\,\}
\]
(the structural--instability locus). Here $S\subset\mathcal{X}$ denotes the set of structurally stable germs, and we write $S^{c}$ for the corresponding structural--instability locus (the ``singular'' set), i.e.\ the set of germs that fail structural stability.

The \emph{codimension of the degeneracy} is the number of independent conditions needed to reach $S^{c}$ from a generic germ, equivalently the minimal number of control parameters in a (topologically/smoothly) versal deformation.

For a one--dimensional reduced germ $g(Z,\bar\mu)=\bar\mu-Z^{2}$, the equilibrium condition is one scalar equation $g=0$; degeneracy adds the single additional scalar constraint $\partial_Z g=0$ at the same point.
At $(Z,\bar\mu)=(0,0)$ we have $\partial_{ZZ}g\neq 0$ and $\partial_{\bar\mu}g\neq 0$ (cf.\ \eqref{cond}), so the unfolding by the single parameter $\bar\mu$ is transverse and versal.
Hence the fold (saddle--node) is of \emph{codimension one}.

\begin{theorem}[Saddle-node for the unfolded Friedmann system]\label{thm:fold-FRW}
Consider the reduced Friedmann density equation with unfolding,
\be\label{eq:Omega-unfolded}
\Omega'=\Omega(\Omega-1)+\nu,\qquad \nu=\sigma/\mu,
\ee
and its normal form in the shifted variable $Z=\tfrac12-\Omega$,
\be\label{eq:SN-normal}
Z'=\bar{\mu}-Z^2,\qquad \bar{\mu}=\tfrac14-\nu.
\ee
Then the following hold:

\begin{enumerate}
\item (\emph{Equilibria and threshold}) There exists a unique saddle-node  at $(Z,\bar{\mu})=(0,0)$, i.e.\ at $(\Omega,\nu)=(\tfrac12,\tfrac14)$. For $\bar{\mu}>0$ ($\nu<\tfrac14$) there are exactly two real equilibria
\be\label{eq:Omega-branches}
Z_{\pm}=\pm\sqrt{\bar{\mu}},\qquad
\Omega_{\pm}=\tfrac12\mp\sqrt{\bar{\mu}} \;=\; \frac12\Bigl(1\mp \sqrt{1-4\nu}\Bigr),
\ee
while for $\bar{\mu}<0$ ($\nu>\tfrac14$) there are no real equilibria.

\item (\emph{Stability}) Along the vertical phase lines of the one-dimensional flow, the lower branch $\Omega_{+}$ is asymptotically stable (future attractor) and the upper branch $\Omega_{-}$ is unstable (past attractor). Equivalently, $Z_{+}$ is stable and $Z_{-}$ unstable.

\item (\emph{Square-root splitting}) Near the threshold, the branch separation scales as
\be
\Delta \Omega \equiv \Omega_{-}-\Omega_{+} = 2\sqrt{\bar{\mu}} = \sqrt{1-4\nu},
\ee
with next correction $O(\bar{\mu}^{3/2})$ in generic perturbations.

\item (\emph{Versality and structural stability}) In a neighbourhood of the saddle-node, any sufficiently small $C^{r}$ ($r\ge 2$) perturbation of the FRW vector field that preserves the one-dimensional centre is locally $C^{1}$-equivalent (under smooth changes of state, parameter, and time) to the normal form $Z'=\bar{\mu}-Z^2$. In particular, the qualitative picture above is \emph{robust}.
\end{enumerate}
\end{theorem}

\begin{proof}[Proof sketch]
Near the dispersive Milne state (EQ-I), the centre manifold is the $\Omega$-axis (the $H$-direction is hyperbolic). The reduced germ has the degenerate form $f(\Omega,\mu)= -\mu\,\Omega+\mu\,\Omega^2+O(3)$, which violates both transversality and nondegeneracy at the origin. Adding the unique (up to smooth equivalence) codimension-one term $\sigma$ produces the versal family $G(\Omega,\mu,\sigma)=-\mu\Omega+\mu\Omega^2+\sigma$. After the time rescaling $T=\mu\tau$ and the reparametrisation $\nu=\sigma/\mu$, one obtains the canonical normal form $Z'=\bar{\mu}-Z^2$ with $(Z,\bar{\mu})=(\tfrac12-\Omega,\tfrac14-\nu)$. Nondegeneracy ($\partial_{ZZ}g\neq 0$) and transversality ($\partial_{\bar{\mu}}g\neq 0$) hold at the saddle-node, whence items (1)--(3). Versality in item (4) follows from standard singularity/bifurcation arguments (finite determinacy and Thom transversality) applied to one-dimensional centre germs.
\end{proof}

The following consequences of Theorem \ref{thm:fold-FRW} now follow.
\begin{corollary}[Classical FRW as the singular limit]\label{cor:nu0}
At $\nu=0$ (i.e.\ $\bar{\mu}=\tfrac14$) the interior branches collapse to the classical boundary equilibria $\Omega=0$ and $\Omega=1$. For $0<\nu<\tfrac14$ the boundary points are replaced by the interior pair $(\Omega_{+},\Omega_{-})$. At $\nu=\tfrac14$ they annihilate at $\Omega=\tfrac12$. For $\nu>\tfrac14$ no equilibria remain.
\end{corollary}

\begin{corollary}[Sign of $q$ at the versal equilibria]\label{cor:q-sign}
At an equilibrium, $q=\tfrac12\,\mu\,\Omega$. Hence:
\begin{itemize}
\item For $\mu>0$ and $\nu<0$ one has $\Omega_{+}<0$, and therefore $q<0$ (accelerating attractor) even without $\mu<0$.
\item For $\nu>0$ one has $\Omega_{+}>0$, so acceleration at $\Omega_{+}$ requires $\mu<0$.
\end{itemize}
\end{corollary}

\begin{remark}[Model independence near the saddle-node]
Items \ref{thm:fold-FRW}(2)--(4) are invariant under smooth redefinitions of $(\Omega,\nu)$ and smooth time rescalings. Thus the phase portrait and the square-root law are properties of the \emph{class} of unfolded FRW systems, not of a particular coordinatisation.
\end{remark}

These computations settle finite determinacy near the saddle-node: in a neighbourhood of $(Z,\bar{\mu})=(0,0)$, any system whose germ reduces to
\be
f(\Omega,\mu)=a_0\mu+a_1\Omega^2+a_2\mu\Omega+a_3\mu^2+O(3),
\ee
(i.e.\ any sufficiently small perturbation of the Friedmann equations) is \emph{qualitatively} equivalent to \eqref{un3} (equivalently, to $\mu\,\Omega(\Omega-1)+\sigma$) under smooth changes of state, parameter and time. Moreover, the unfolded family \eqref{un3} has \emph{codimension one}: precisely one auxiliary parameter is needed. Hence the deformation is \emph{generic}, and by Thom's transversality theorem the unfolding is \emph{versal}.

\section{Cosmic acceleration at the saddle--node threshold}\label{sec:accel}

In this section we interpret the unfolding parameters in physical terms (via a modified continuity/entropy balance), identify the accelerating future attractor beyond the saddle-node, and note novel observational identities that connect \((q_{0},\Omega_{0},H_{0})\) to \((\sigma_{0},\nu_{0},\bar\mu_{0})\).

\subsection{Unfolding parameter and entropy production}\label{subsec:entropy}

The versal deformation implies a modified continuity equation
\be\label{eq:mod-cont}
\dot\rho+3H(\rho+p)=3\sigma H^{3},
\ee
which reduces to the adiabatic law when \(\sigma=0\). Using the energy-volume definition, \(E=\rho V\), \(V=2\pi^{2}a^{3}\), and \(dE=T\,dS - p\,dV\), one obtains the entropy balance,
\be\label{eq:entropy-balance}
\frac{\dot S}{V}=\frac{3\sigma}{T}\,H^{3}.
\ee
Thus \(\sigma\) parametrizes departures from adiabaticity: for an expanding universe (\(H>0\)), \(\sigma>0\) gives entropy production \((\dot S>0)\) and \(\sigma<0\) entropy decrease.

It is also convenient to package \eqref{eq:mod-cont} as a standard-looking continuity law with an effective equation-of-state parameter \(w_{\rm eff}\):
\be\label{eq:w-eff}
\dot\rho+3H(1+w_{\rm eff})\rho=0,
\qquad
w_{\rm eff}=w-\frac{\sigma}{3\Omega},
\qquad \Omega=\frac{\rho}{3H^{2}}.
\ee
This identity will be useful below.

\subsection{\texorpdfstring{The accelerating future attractor $\Omega_{+}$}
                         {The accelerating future attractor Omega+}}
\label{subsec:acc-attractor}
Crossing the saddle-node point corresponds to \(\bar\mu=0\), i.e.\ \(\nu=\tfrac14\), equivalently,
\be\label{eq:fold-line}
\sigma=\frac{\mu}{4}.
\ee
In the regime \(\nu<\tfrac14\) one has the pair \eqref{eq:Omega-branches} with \(\Omega_{+}\) future-stable and \(\Omega_{-}\) past-stable. The following gives the acceleration condition on the \(\Omega_{+}\) branch. By \eqref{eq:Omega-branches}, \(\Omega_{+}\le \tfrac12\) with equality only at the saddle-node point, and  moreover, \(\Omega_{+}<0\) iff \(\nu<0\). Then the following Proposition follows from \eqref{eq:q-eq} and \eqref{eq:entropy-balance}.

\begin{proposition}[Acceleration beyond the saddle-node point]\label{prop:acc}
Let \(\mu=3\gamma-2\) be fixed and consider the unfolded density dynamics \eqref{eq:Omega-unfolded} with \(\nu<\tfrac14\). On the future attractor \(\Omega_{+}\), the deceleration parameter \eqref{eq:q-eq} is
\be
q_{+}=\frac{1}{2}\mu\,\Omega_{+}.
\ee
Hence:
\begin{enumerate}\setlength\itemsep{2pt}
\item If \(\mu>0\) and \(\nu<0\) (so \(\sigma/\mu<0\)), then \(\Omega_{+}<0\) and \(q_{+}<0\) (late-time acceleration). In this case \(\sigma<0\) and \(\dot S<0\) by \eqref{eq:entropy-balance}.
\item If \(\mu<0\) and \(0<\nu<\tfrac14\), then \(\Omega_{+}>0\) and again \(q_{+}<0\). Here \(\sigma>0\) and \(\dot S>0\).
\end{enumerate}
\end{proposition}
Introducing redshift dependence, we set \(a=(1+z)^{-1}\) and the \emph{normalized expansion rate} \(E(z)=H(z)/H_{0}\) (so that $E(0)=1$). The \emph{transition redshift} \(z_t\) is defined by \(q(z_t)=0\) (onset of acceleration).
There are various further interesting forms that the  basic system \eqref{bas-1} may take,  identities for the background quantities  $(H(z),\Omega(z))$ or  $(E(z),\Omega(z))$.
\begin{proposition}[Background identities]\label{prop:qOmega}
The Raychaudhuri equation \eqref{bas-1}b reads,
\be\label{ray-z}
q(z)= -1+(1+z)\,\frac{dH/dz}{H(z)}
\ee
or, in terms of the function \(E(z)\), this becomes,
\be \label{E-z}
(1+z)\,\frac{dE}{dz}\;=\; E(z)(1+q).
\ee
Further the $\Omega$-equation \eqref{bas-1}a becomes,
\be
(1+z)\,\frac{d\Omega}{dz}\;=\;-\left( \Omega(\Omega-1) +\nu\right).
\ee
\end{proposition}
\noindent
From the background identity $q(z)=\tfrac{\mu}{2}\,\Omega(z)$ it follows that
\(
q(z)<0 \;\Longleftrightarrow\; \mu\,\Omega(z)<0.
\)
With $\mu>0$ (ordinary matter; e.g.\ dust has $\mu=1$), acceleration therefore lies on the $\Omega(z)<0$ branch generated by the bounded unfolding. No separate adiabatic dark–energy component with $w<-1/3$ is invoked: in this precise sense our mechanism is \emph{not} dark energy, but a structural consequence of the organizing center and its bounded versal unfolding.

Also we can rephrase the attractor properties in terms of $z$: at a fixed point $\Omega_\star$ of $(1+z)\Omega'=-\bigl(\Omega(\Omega-1)+\nu\bigr)$, the $z$-flow linearization is $\delta\Omega' = -\,(2\Omega_\star-1)\,\delta\Omega/(1+z)$, so the branch with $\Omega_\star<\tfrac12$ is attractive (forward in $z$) while $\Omega_\star>\tfrac12$ is repulsive. Next, we give two results about asymptotics and age:
For the high--redshift regime we choose the decelerating branch: with $\mu>0$ and $\Omega(z)>0$ at $z\gg1$, the identity
$q(z)=\tfrac{\mu}{2}\,\Omega(z)$ implies $q(z)>0$, i.e.\ deceleration at early times.
Moreover, the present age is finite:
\begin{equation}
t_0 \;=\; H_0^{-1}\int_{0}^{\infty}\frac{dz}{(1+z)\,E(z)} \;<\;\infty,
\end{equation}
because the background relation $(1+z)\,E'/E = 1 + \tfrac{\mu}{2}\,\Omega(z)$ with bounded, positive $\Omega(z)$ at $z\gg1$ yields
$E(z)\sim (1+z)^{\,1+\frac{\mu}{2}\Omega_\infty}$ (for some $\Omega_\infty>0$), so the integrand scales as $(1+z)^{-2-\frac{\mu}{2}\Omega_\infty}$ and is integrable at infinity.

To see the practical importance of these forms, we derive below two more results, one for the observational determination of $E(z)$ in terms of $\Omega_0,q_0$, and another in terms of $H_0,q_0$, and the present jerk $j_0$ (present curvature of $E(z)$, i.e., $d^2E/dz^2|_0$). Both of these results are used in an essential way in our companion paper \cite{cot26a}, where we map a choice of unfolding history $\nu(z)=\sigma(z)/\mu$ to the observable expansion $E(z)\equiv H(z)/H_0$, which in turn constrains $\bar\mu(z)=\tfrac14-\nu(z)$. We write derivatives with respect to $N\equiv\ln a$ (so $X'\equiv \mathrm{d}X/\mathrm{d}N$), with $a=(1+z)^{-1}$, and use the unfolded continuity equation and Raychaudhuri one to obtain the background system,
\begin{align}
\Omega' &= \Omega(\Omega-1)+\nu(N),
&& \nu(N)\equiv\frac{\sigma(N)}{\mu},
\label{appB:OmegaPrime}\\
\frac{\mathrm{d}\ln E}{\mathrm{d}N} &= -\Bigl(1+\frac{\mu}{2}\,\Omega\Bigr),
&&\Longleftrightarrow\quad
(1+z)\,\frac{1}{E}\frac{\mathrm{d}E}{\mathrm{d}z}
= 1+\frac{\mu}{2}\,\Omega(z).
\label{appB:EPrime}
\end{align}

Throughout we assume the unfolding control $\nu(z)$ is a bounded, smooth function of redshift (equivalently of the scale factor $a=(1+z)^{-1}$),
$\nu:\ [0,\infty)\to\mathbb{R},\, \nu\in C^1,$ with $|\nu(z)|<\infty,\, \forall z\ge 0$.
This guarantees a well-behaved background map (no blow-ups in $E(z)$, finite age integral, and a decelerating high-$z$ branch consistent with our identities).
We do not fix a specific parametric form for $\nu$ here; concrete, smooth bounded choices and their impact on fits are developed in the observational companion paper \cite{cot26a}.

Initial conditions at $z=0$ are $E(0)=1$ and $\Omega(0)=\Omega_0$; if a present deceleration $q_0$ is adopted, one may set
\be
\Omega_0=\frac{2q_0}{\mu}.
\ee
Near the saddle-node point, the distance parameter is $\bar\mu(z)=\tfrac14-\nu(z)$. Thus, \emph{choosing a parameterisation for $\nu(z)$ (equivalently $\bar\mu(z)$) and integrating \eqref{appB:OmegaPrime}–\eqref{appB:EPrime} yields a predictive $E(z)$ to confront with low-redshift data}, with the early-time prior $\sigma(z_{\mathrm{rec}})\simeq 0$ enforcing the adiabatic regime. By finite determinacy and versality (Remark~3.1), the germ-level predictions are invariant under smooth reparametrisations; observational analyses may therefore adopt any bounded, smooth parameterisation of $\nu(z)$ and fit its (few) coefficients to data.

The second result introduces exploits the idea of an ``optional low-$z$ anchor'', by which  we mean using today’s jerk $j_{0}$ to match the second derivative of $E(z)$ at $z=0$ (with $E''(0)=j_{0}-q_{0}^{2}$), alongside $H_{0}$ and $q_{0}$ which fix the value and slope. This tightens the background fit without introducing new dynamics.
In our background system,
\be
q(z)=\frac{\mu}{2}\,\Omega(z),\qquad
j(z)=q(z)+2q(z)^2+(1+z)\,\frac{{\rm d}q}{{\rm d}z},
\ee
with
\be
(1+z)\,\frac{{\rm d}\Omega}{{\rm d}z}
= -\big[\Omega(\Omega-1)+\nu(z)\big],
\qquad
(1+z)\,\frac{{\rm d}q}{{\rm d}z}
= -\frac{\mu}{2}\big[\Omega(\Omega-1)+\nu(z)\big].
\ee
Evaluated at $z=0$ this gives
\be
j_0 \;=\; q_0 + 2q_0^2 \;-\; \frac{\mu}{2}\,\Big[\Omega_0(\Omega_0-1)+\nu_0\Big],
\qquad \big(q_0=\tfrac{\mu}{2}\Omega_0\big),
\ee
and the identities
\be
\left.\frac{1}{E}\frac{{\rm d}E}{{\rm d}z}\right|_{0}=1+q_0,
\qquad
\left.\frac{{\rm d}^2E}{{\rm d}z^2}\right|_{0}=j_0-q_0^2,
\ee
so that we arrive at the final result,
\be
E(z)=1+(1+q_0)z+\tfrac12\,(j_0-q_0^2)\,z^2+O(z^3).
\ee
Including $j_0$ as a low-$z$ anchor means the fit to $\nu(z)$ (equivalently $\bar\mu(z)$) also matches this second-order behaviour near $z=0$.

\subsection{Macroscopic representations of the source term}\label{subsec:effective}

The non-adiabatic source \(3\sigma H^{3}\) in \eqref{eq:mod-cont} admits several \emph{macroscopic representations}. These are \emph{algebraically equivalent} at the homogeneous background level—they reproduce the same \(\rho(a)\) once the identifications below are made—but they encode different constitutive physics and generally lead to \emph{inequivalent} predictions once transport/perturbations are specified.

\paragraph{(i) Bulk viscosity / creation pressure.}
We package non-adiabatic effects via an effective pressure
\begin{equation}\label{eq:pEffDef}
p_{\mathrm{eff}} \equiv p - \Pi,
\end{equation}
so that the continuity equation reads
\begin{equation}\label{eq:Pi}
\dot{\rho}+3H\big(\rho+p_{\mathrm{eff}}\big)=0.
\end{equation}
In our unfolding background we encode the source by
\begin{equation}\label{eq:Qsigma}
\dot{\rho}+3H(\rho+p)=3\,\sigma\,H^{3},
\end{equation}
hence, for $H>0$,
\begin{equation}\label{eq:PiModel}
\Pi=\sigma\,H^{2}.
\end{equation}
(When the generic source $\mathcal Q$ is introduced later, \eqref{eq:Qsigma} corresponds to $\mathcal Q(t)=3\sigma H^3$.)
Then, for the $\sigma<0$, $\mu>0$ regime in the flat case, with $H^{2}=\rho/3$, one has,
\[
\Pi=\sigma H^{2}=\tfrac{\sigma}{3}\,\rho, \qquad p=w\rho=\tfrac{\mu-1}{3}\,\rho.
\]
Hence,
\[
\frac{|\Pi|}{p}=\frac{(\!-\sigma/3)\rho}{\bigl((\mu-1)/3\bigr)\rho}
=\frac{|\sigma|}{\mu-1}\quad(\mu\neq 1),
\]
and for $\sigma<0$,
\(
p_{\rm eff}=p-\Pi>p,
\)
i.e.\ the effective pressure is \emph{larger} than the bare $p$.
For dust ($\mu=1$) one has $p=0$ and $\Pi/\rho=\sigma/3$, so $p_{\rm eff}=-\Pi>0$ if $\sigma<0$.
We do not impose a global sign on $\Pi$ here; the comparative thresholds remain
\(
p_{\rm eff}<0 \iff \Pi>p,\,\,
p_{\rm eff}=0 \iff \Pi=p,\,\,
p_{\rm eff}>0 \iff \Pi<p.
\)
Then for compatibility with the accelerating branch, we note that the effective–pressure bookkeeping \eqref{eq:pEffDef} with $\Pi$ defined by \eqref{eq:PiModel} does not determine acceleration; it merely repackages the source term. Our accelerating solutions use the background criterion established above (with $\mu>0$ on the unfolding branch $\Omega(z)<0$). Hence both signs of $\sigma$ (and thus $\Pi=\sigma H^{2}$) are compatible with acceleration. In particular, for $\sigma<0$ one has $p_{\rm eff}=p-\Pi>p$ (dust: $p=0\Rightarrow p_{\rm eff}>0$), yet the $\Omega(z)<0$ branch still accelerates by the background criterion.

\paragraph{(ii) Particle creation in an open system.}
Allow a homogeneous creation rate \(\Gamma\) via
\be\label{eq:Gamma}
\dot\rho+3H(\rho+p)=(\rho+p)\,\Gamma,
\qquad
\Gamma=\frac{3\sigma H^{3}}{\rho+p}.
\ee
Then \(\sigma>0\Rightarrow\Gamma>0\) (net creation, \(\dot S>0\)), and \(\sigma<0\Rightarrow\Gamma<0\) (net annihilation/order, \(\dot S<0\)).

\paragraph{\texorpdfstring{(iii) Bookkeeping via an entropy-sector fraction $\Omega_\sigma$.}
                          {(iii) Bookkeeping via an entropy-sector fraction Omega_sigma.}}
Represent the source as a homogeneous sector exchanging energy with the baseline fluid:
\begin{align}
\dot\rho+3H(1+w)\rho&=-Q, &
\dot\rho_\sigma+3H(1+w_\sigma)\rho_\sigma&=+Q,
\end{align}
with exchange \(Q\equiv-3\sigma H^{3}\) and \(\Omega_\sigma\equiv\rho_\sigma/(3H^2)\). The total then obeys the standard continuity law with the effective equation of state \(w_{\rm eff}=w-\sigma/(3\Omega)\) in \eqref{eq:w-eff}. Two closures are  especially transparent:
\begin{itemize}
\item \emph{Creation–pressure limit} ($w_\sigma\simeq -1$): the sourced component is vacuum–like. For $Q>0$ the vacuum–like density grows and drives $w_{\rm eff}$ more negative, favouring acceleration at the background level.
\item \emph{Tracking limit} ($w_\sigma=w$): the sourced component mimics the primary fluid. Then
$\dot\rho_{\rm tot}+3H(1+w)\rho_{\rm tot}=0$ and $w_{\rm eff}=w$, so the background expansion $H(z)$ is unchanged; the source only reshuffles $\Omega$ vs.\ $\Omega_\sigma$.
\end{itemize}
We will use these limits as convenient end–members when interpreting the unfolding source at the background level.

\begin{remark}[Scope of equivalence]
At the \emph{background} level—and with \(\Omega_\sigma\) kept homogeneous—the three representations are dynamically equivalent once one identifies
\(
\Pi=-\sigma H^{2},\
\Gamma=3\sigma H^{3}/(\rho+p),\
Q=-3\sigma H^{3}.
\)
They are \emph{not} physically identical: bulk-viscous fluids entail transport coefficients and relaxation times (Eckart/Israel–Stewart), creation models specify microphysics of \(\Gamma\), and two-fluid exchanges require choices of \(w_\sigma\) and sound speed. Differences become observable once perturbations/transport are included. See, e.g., \cite{Eckart,IsraelStewart,Maartens,Prigogine1988,CalvaoLimaWaga,Lima1996,Zimdahl}. At linear order in perturbations these representations generically differ through their sound speeds, viscosities, relaxation times, and source clustering; hence observational constraints on growth can, in principle, discriminate between them.

\end{remark}
\begin{example}[Constant $\nu$]\label{ex:const-nu}
Let $\nu=\mathrm{const}$, $\Delta\equiv\sqrt{1-4\nu}$, and $\Omega_\pm=\tfrac12(1\mp\Delta)$.
From \eqref{appB:OmegaPrime} one has
$R(N)\equiv\frac{\Omega-\Omega_+}{\Omega-\Omega_-}=R_0\,e^{-\Delta(N-N_0)}$, with
$R_0=\frac{\Omega_0-\Omega_+}{\Omega_0-\Omega_-}$.
Using \eqref{appB:EPrime} and setting $N_0=0$ ($a_0=1$), the expansion reads
\be
E(a)=\mathrm{const}\times a^{-\left[\,1+\frac{\mu}{4}(1-\Delta)\right]}\,\bigl(1-C\,a^{-\Delta}\bigr)^{\mu/2},
\qquad C\equiv R_0.
\ee
\emph{Remark.} Typically $0<C<1$ when the initial state lies between the branches ($\Omega_+<\Omega_0<\Omega_-$); other choices of $\Omega_0$ are allowed (then $C$ may be $\le0$ or $>1$), with sign absorbed into the overall constant.

\emph{Relaxation.} The approach to the stable branch $\Omega_+$ is exponential in e-folds with rate $\Delta$ (since $f'(\Omega_+)=2\Omega_+-1=-\Delta$): $N_{\rm relax}\!\simeq\!1/\Delta$, $N_{1/2}\!=\!\ln2/\Delta$; thus $\nu=0$ sets the benchmark ($\Delta=1$), $\nu\to\tfrac14^-$ slows ($\Delta\to0^+$), and $\nu<0$ speeds up ($\Delta>1$).

\emph{Asymptotics.} As $a\to\infty$, $(1-C\,a^{-\Delta})^{\mu/2}\to1$ and we find the \emph{Versal–Milne law},
\[
E(a)\sim \mathrm{const}\times a^{-\left[\,1+\frac{\mu}{4}(1-\Delta)\right]},
\]
i.e. a Milne baseline $a^{-1}$ tilted by the versal control. For $\nu=0$ ($\Delta=1$) this gives $E\sim a^{-1}$, the Milne-like late-time behaviour corresponding to the \emph{stable} branch $\Omega_+=0$. Recovering the flat single-fluid FRW power law $E\sim a^{-(1+\mu/2)}$ at $\nu=0$ requires sitting \emph{exactly} on the \emph{unstable} branch $\Omega_-\equiv1$; any perturbation relaxes to Milne and restores $E\sim a^{-1}$.

\emph{Observable ratio (choosing Versal–Milne as the late-time standard).}
Normalising at $z=0$ ($E(1)=1$) and taking $E_{\rm std}(a)=a^{-1}$,
\be
\frac{E_{\rm unfold}(a)}{E_{\rm std}(a)}
= a^{-\frac{\mu}{4}(1-\Delta)}\left[\frac{1-C\,a^{-\Delta}}{1-C}\right]^{\mu/2}.
\ee
Illustration (dust, $\mu=1$): with $\nu=-0.02$ ($\Delta\!\approx\!1.039$) and $C=0.03$,
$E(z{=}1)/E_{\rm std}\!\approx\!0.977$ and $E(z{=}0.5)/E_{\rm std}\!\approx\!0.988$;
for $\nu=-0.05$ ($\Delta\!\approx\!1.095$) and the same $C$,
$E(z{=}1)/E_{\rm std}\!\approx\!0.966$ and $E(z{=}0.5)/E_{\rm std}\!\approx\!0.982$.
Thus the unfolding predicts $\sim\!1$--$3\%$ lower $H$ around $z\!\sim\!0.5$--$1$ for modest negative $\nu$ (stronger for larger $|\nu|$ or larger $C$).

\emph{Why $\rho(a)$ probes non-adiabaticity.} Since $\rho(a)=3H_0^2\,E(a)^2\,\Omega(a)$ and both $E(a)$ and $\Omega(a)$ acquire $\nu$-dependent, non-power-law factors even for constant $\nu$, $\rho(a)$ deviates from the linearised adiabatic scaling $\rho\propto a^{-3(1+w)}$. Departures of $\log\rho$ vs.\ $\log a$ from a straight line thus provide a direct, background-level probe of non-adiabaticity.
\end{example}

While a constant $\nu$ can reproduce percent-level features of $H(z)$ and cleanly illustrates the mechanism, it is not a free-for-all: the background identities and the high-$z$ adiabatic prior $\nu\!\to\!0$ strongly restrict admissible histories. In the companion observational paper \cite{cot26a} we therefore adopt bounded, smooth forms with few parameters, and use $(H_0,q_0,j_0)$ together with $H(z)$ to pin them down.

\section{Discussion}\label{sec:discussion}

We have shown that (i) the FRW system admits a codimension-one versal unfolding whose reduced normal form is the saddle–node \eqref{un3}; (ii) the unfolding parameter $\nu=\sigma/\mu$ controls the creation/annihilation of equilibria $\Omega_\pm$ given by \eqref{eq:Omega-branches}, with $\Omega_+$ future-stable and $\Omega_-$ past-stable; (iii) $\sigma$ has a clear thermodynamic meaning via the first-law balance \eqref{eq:entropy-balance}, where \(3\sigma H^{3}\) sources \(T\dot S\), and it induces a clean background-level shift of the equation of state through \eqref{eq:w-eff}, \(w_{\rm eff}=w-\sigma/(3\Omega)\)
; and (iv) for $\mu>0$ the future attractor $\Omega_+$ leads to acceleration when $\nu<0$ (hence $\sigma<0$), offering a purely dynamical route to late acceleration without a cosmological constant.

Our analysis is restricted to the \emph{background} (homogeneous/isotropic) dynamics and to a \emph{neighbourhood of the saddle-node point} in parameter space.
By ``background'' we mean that we do not model linear or non-linear perturbations (clustering, lensing, or CMB anisotropies); by ``near the saddle-node'' we mean that we truncate to the terms that organise the saddle–node normal form, so higher-order couplings and additional degrees of freedom are neglected away from the saddle-node point.
Finally, the results are sensitive to adding extra degrees of freedom (cosmological constant, scalar fields, higher-order curvature), which enlarge the organising centre and can introduce additional bifurcations beyond the saddle-node.

The explanation advanced here is independent of a cosmological constant in the sense that $\nu$ is unrelated to $\Lambda$ and the governing equations are inequivalent to, and simpler than, the Friedmann–Lema\^itre system. Including $\Lambda$ unfolds the de Sitter and Einstein-static equilibria (impossible in the pure Friedmann case), and leads to higher bifurcations and mode interactions \cite{cot24a}. Thus while the present model is structurally simple, a full bifurcation analysis of extensions (with $\Lambda$, scalar fields, or higher-curvature terms) can reveal additional organising modes.

Because acceleration hinges on $\nu\neq 0$, it correlates with nonzero and evolving entropy through \eqref{eq:entropy-balance}–\eqref{eq:w-eff}, i.e.\ an out-of-equilibrium evolution governed by the modified continuity law. This is a manifestation of a broader \emph{gravitational self-organisation}: parameter-dependent creation of new structures (the equilibria $\pm\sqrt{\bar{\mu}}$) as a direct consequence of unfolding the Milne and flat equilibria. The system tends toward a precarious state that disappears as parameters vary: the $\Omega_\pm$ exist only for $\bar{\mu}>0$ (i.e.\ $\nu<1/4$), collide at $\bar{\mu}=0$, and are absent for $\bar{\mu}<0$. Nevertheless, the global organisation by the saddle–node \eqref{un3} (or \eqref{un1b}) is structurally stable, i.e.\ persists under small perturbations of the variables.

This interpretation is reminiscent of self-organised criticality in complex systems \cite{bak1,bak,bar07}, with a square-root scaling law for the branch splitting, $\Delta\Omega\propto\sqrt{\bar{\mu}}$ (cf.\ \cite{cot23a}, Sec.~8.2.2). The gravitational context suggests a broader universality of self-organisation via bifurcations in general relativity.

The unfolded Friedmann normal form $\Omega'=\Omega(\Omega-1)+\nu$ (with $Z=\tfrac{1}{2}-\Omega$ and $\bar\mu=\tfrac{1}{4}-\nu$) suggests a global scenario if the control $\nu(a)=\sigma(a)/\mu(a)$ drifts monotonically from $\nu>\tfrac{1}{4}$ to $\nu<0$. For $\nu>\tfrac{1}{4}$ (pre-birth) no equilibria exist and the flow is monotone (Fig.~1). At $\nu=\tfrac{1}{4}$ a saddle--node is born at $\Omega=\tfrac{1}{2}$, with $d\Omega_\pm/d\nu\to\pm\infty$. For $0<\nu<\tfrac{1}{4}$ the pair $\Omega_\pm=\tfrac{1}{2}\mp\sqrt{\bar\mu}$ emerges; choosing $\mu<0$ yields $q_+=\tfrac{\mu}{2}\Omega_+<0$ (acceleration) and $\sigma>0$ (entropy production), providing an inflation-like, SEC-violating early phase within the same versal family. The adiabatic touchpoint $\nu=0$ recovers the classical $\Omega=0,1$ equilibria. Subsequently, for $\nu<0$ with $\mu>0$, the stable branch moves to $\Omega_+<0$, again giving $q_+<0$ (late acceleration) but now with $\sigma<0$. The “entropy decrease’’ here refers to the fluid bookkeeping ($\dot S/V=(3\sigma/T)H^3$); a generalized entropy including horizon/gravitational sectors can still increase. This scenario is structurally natural near the saddle--node and falsifiable through the sign pattern $\mathrm{sgn}(\mu\,\Omega)$, high-$z$ deceleration, and finite age.

The square-root splitting $\Delta\Omega=2\sqrt{\bar\mu}$ quantifies the geometric distance between equilibria in the unfolded phase portrait, but it does not determine proper time. Cosmological duration along a history depends on the control trajectory $\nu(a)$ and on the relation $dt/d\tau=1/H$ (with $T=\mu\tau$). In particular, at $\nu=0$ the endpoints $\Omega=0,1$ are equilibria, so the approach times diverge in the reduced flow; hence no direct identification between $\Delta\Omega$ and lifetime is warranted.

At the background level, a strictly constant control $\nu$ already captures the essential saddle–node geometry and yields closed forms for $\Omega(a)$ and $E(a)$ with percent-level deviations from the adiabatic scaling. However, matching \emph{all} low-$z$ anchors simultaneously (e.g., $H_0,q_0,j_0$) and the full $H(z)$ curve typically prefers a \emph{bounded, smooth} history $\nu(z)$ that tends to $0$ at high redshift (adiabatic limit). In practice, we adopt minimal priors (boundedness, $C^1$ in $a$, $\nu(z\!\gg\!1)\!\to\!0$) and parameterize, e.g., $\nu(a)=\nu_0+\nu_1\,g(a)$ with the canonical bounded shape $g(a)=3a^2-2a^3$. This retains the codimension-one mechanism while allowing the small departures from adiabaticity that the data may require \cite{cot26a}.

In the saddle--node picture developed here, the center--manifold reduction yields the germ
\(
Z'=\bar\mu-Z^{2},\quad \bar\mu=\tfrac14-\nu,
\)
so that the relaxation span near the fold scales as \(N_{\rm relax}\propto 1/\Delta\) with \(\Delta=\sqrt{1-4\nu(z_t)}\).
Because \(\mathrm{d}t=\mathrm{d}N/H\) and \(H\) varies slowly across the bottleneck, the look--back time through a narrow pre--transition window \([z_t-\delta,\,z_t]\) is strictly larger than in any smooth comparator calibrated to the same \((H_0,z_t)\):
\(
\Delta t_{\rm ghost}\sim\frac{C}{\Delta\,H_t}>0,\;
H_t:=H(z_t),\; C=\mathcal O(1).
\)
Two ancillary kinematical consequences follow: (i) an \emph{asymmetric} passage of \(q(z)\) through zero (slow approach to \(q=0\), faster departure), and (ii) a small positive offset between the inflection redshift \(z_{\rm inf}\) defined by \(E''(z_{\rm inf})=0\) and the acceleration onset \(z_t\) defined by \(q(z_t)=0\).
These provide a compact, falsifiable discriminator of the mechanism that can be checked with CC reconstructions of \(H(z)\) in a windowed analysis, without committing to a global fit.\footnote{CC denotes \emph{cosmic chronometer} determinations of \(H(z)\) from differential galaxy ages. A practical choice is a symmetric window \([z_t-\delta,z_t+\delta]\) with \(\delta\simeq0.1\) and residuals relative to a \(\Lambda\)CDM comparator matched at \((H_0,z_t)\) (or, alternatively, at \((H_0,q_0)\)).}

The framework is tightly testable using only background identities: (i) \emph{sign test} $\operatorname{sgn}(\mu\,\Omega)$—for $\mu>0$ late acceleration requires $\nu<0$ on the $\Omega_+$ branch; (ii) \emph{transition width}—the relaxation rate $\Delta(z)=\sqrt{1-4\nu(z)}$ controls the sharpness of the decel$\to$accel transition in e-folds; (iii) \emph{asymptotic slope}—as $a\to\infty$, $E(a)\!\sim\!\text{const}\times a^{-\left[1+\frac{\mu}{4}\big(1-\Delta_\infty\big)\right]}$, so the exponent shift $\delta p=\frac{\mu}{4}(1-\Delta_\infty)$ is directly testable; (iv) \emph{direct reconstruction}—from $E,E',E''$ one obtains $\nu(z)$ at low $z$ (optionally anchored by $q_0,j_0$). These observables constrain $(\nu_0,\nu_1,\ldots)$ without invoking a cosmological constant; the companion paper develops the corresponding fits and priors \cite{cot26a}.
\addcontentsline{toc}{section}{Acknowledgments}
\section*{Acknowledgments}
This research  was funded by RUDN University,  scientific project number FSSF-2023-0003.


\addcontentsline{toc}{section}{References}
\begin{thebibliography}{99}

\bibitem{cot24c}
S. Cotsakis, \emph{Structural stability and general relativity}, Universe 11(7) (2025) 209; arXiv:2412.04283.

\bibitem{cot23a}
S. Cotsakis, \emph{Bifurcation diagrams for spacetime singularities and black holes}, Eur. Phys. J. C 84 (2024) 35; arXiv:2311.16000.

\bibitem{cot23b}
S. Cotsakis, \emph{The crease flow on null hypersurfaces}, Eur. Phys. J. C 84 (2024) 391; arXiv:2312.08023.

\bibitem{cot24a}
S. Cotsakis, \emph{Friedmann–Lema\^itre universes and their metamorphoses}, Eur. Phys. J. C 85 (2025) 579; arXiv:2411.17286.

\bibitem{a1}
Supernova Search Team Collaboration, A. G. Riess et al., \emph{Observational evidence from supernovae for an accelerating universe and a cosmological constant}, Astron. J. 116 (1998) 1009–1038; arXiv:astro-ph/9805201.

\bibitem{a2}
Supernova Cosmology Project Collaboration, S. Perlmutter et al., \emph{Measurements of $\Omega$ and $\Lambda$ from 42 high-redshift supernovae}, Astrophys. J. 517 (1999) 565–586; arXiv:astro-ph/9812133.

\bibitem{a3}
E. Abdalla et al., \emph{Cosmology Intertwined: A review of the particle physics, astrophysics, and cosmology associated with the cosmological tensions and anomalies}, J. High En. Astrophys. 2204 (2022) 002; arXiv:2203.06142.

\bibitem{d1}
E. J. Copeland, M. Sami, and S. Tsujikawa, \emph{Dynamics of dark energy}, Int. J. Mod. Phys. D 15 (2006) 1753; arXiv:hep-th/0603057.

\bibitem{d2}
L. Amendola and S. Tsujikawa, \emph{Dark Energy: Theory and Observations} (Cambridge University Press, 2010).

\bibitem{i1}
G. F. R. Ellis and T. Buchert, \emph{The universe seen at different scales}, Phys. Lett. A 347 (2005) 38–46; arXiv:gr-qc/0506106.

\bibitem{i2}
T. Buchert, \emph{Dark energy from structure: A status report}, Gen. Relativ. Gravit. 40 (2008) 467–527; arXiv:0707.2153.

\bibitem{q1}
B. Ratra and P. J. E. Peebles, \emph{Cosmological consequences of a rolling homogeneous scalar field}, Rev. Mod. Phys. 75 (2003) 559.

\bibitem{q2}
S. W. Weinberg, \emph{Cosmology} (Oxford University Press, 2008), Sec. 1.12.

\bibitem{m1}
S. Capozziello and M. De Laurentis, \emph{Extended theories of gravity}, Phys. Rep. 509 (2011) 167–321; arXiv:1108.6266.

\bibitem{m2}
T. Clifton, P. G. Ferreira, A. Padilla, and C. Skordis, \emph{Modified gravity and cosmology}, Phys. Rep. 513 (2012) 1–189; arXiv:1106.2476.

\bibitem{cot23}
S. Cotsakis, \emph{Dispersive Friedmann universes and synchronization}, Gen. Relativ. Gravit. 55 (2023) 61; arXiv:2208.07892.

\bibitem{cot26a} S. Cotsakis, \emph{Cosmic acceleration as a saddle-node bifurcation: expansion-history fits and falsifiable predictions}, preprint.

\bibitem{we97}
J. Wainwright and G. F. R. Ellis, \emph{Dynamical Systems in Cosmology} (Cambridge University Press, 1997).

\bibitem{weinberg1}
S. W. Weinberg, \emph{Gravitation and Cosmology} (John Wiley \& Sons, 1972).

\bibitem{hob}
M. P. Hobson, G. Efstathiou, and A. N. Lasenby, \emph{General Relativity: An Introduction for Physicists} (Cambridge University Press, 2006).

\bibitem{Eckart}
C. Eckart, \emph{The thermodynamics of irreversible processes. III.}, Phys. Rev. 58 (1940) 919.

\bibitem{IsraelStewart}
W. Israel and J. M. Stewart, \emph{Transient relativistic thermodynamics and kinetic theory}, Ann. Phys. 118 (1979) 341.

\bibitem{Maartens}
R. Maartens, \emph{Dissipative cosmology}, Class. Quantum Grav. 12 (1995) 1455; arXiv:astro-ph/9609119.

\bibitem{Prigogine1988}
I. Prigogine, J. Geheniau, E. Gunzig, and P. Nardone, \emph{Thermodynamics and cosmology}, Gen. Relativ. Gravit. 21 (1989) 767.

\bibitem{CalvaoLimaWaga}
M. O. Calv\~{a}o, J. A. S. Lima, and I. Waga, \emph{On the thermodynamics of matter creation in cosmology}, Phys. Lett. A 162 (1992) 223.

\bibitem{Lima1996}
J. A. S. Lima, \emph{Thermodynamics of decaying vacuum cosmologies}, Phys. Rev. D 54 (1996) 2571; arXiv:gr-qc/9605055.

\bibitem{Zimdahl}
W. Zimdahl, \emph{Bulk viscous cosmology}, Phys. Rev. D 53 (1996) 5483; and W. Zimdahl, D. Pavón, \emph{Interacting viscous fluid in cosmology}, Phys. Lett. A 176 (1993) 57.


\bibitem{ChevallierPolarski2001}
M. Chevallier and D. Polarski,
\emph{Accelerating universes with scaling dark matter},
Int. J. Mod. Phys. D 10 (2001) 213–223; arXiv:gr-qc/0009008.

\bibitem{Linder2003}
E. V. Linder,
\emph{Exploring the Expansion History of the Universe},
Phys. Rev. Lett. 90 (2003) 091301; arXiv:astro-ph/0208512.

\bibitem{bak1}P. Bak, C. Tang, K. Wiesenfeld, \emph{Self-Organized Criticality: An Explanation of $1/f$ Noise}, Phys. Rev. Lett. 59 (1987) 381–384.

\bibitem{bak}
P. Bak, \emph{How Nature Works: The Science of Self-Organized Criticality} (Oxford University Press, 2006).

\bibitem{bar07}
J. D. Barrow, \emph{New Theories of Everything: The Quest for Ultimate Explanation} (Oxford University Press, 2007).

\end{thebibliography}
\end{document}